%% file: maxCut.tex
\RequirePackage{amsmath}

\documentclass{llncs}
\usepackage{llncsdoc}

\usepackage{amsmath}  
\usepackage[utf8]{inputenc} 
\usepackage[main=english]{babel} 
\usepackage{tikz} 
\usetikzlibrary{automata,arrows,backgrounds,decorations.markings,decorations.pathmorphing,decorations.pathreplacing,fit,positioning,shadows,shapes,shapes.geometric} 
\usepackage{xspace}
\usepackage[unicode,pdfmenubar,linktoc=all,hidelinks,bookmarks]{hyperref} 
\usepackage[margin=0pt,font=small,labelfont=bf]{caption} 
\usepackage{subcaption} 
\usepackage{enumerate} 

\usepackage[plain,chapter]{algorithm}
\usepackage{algorithmic}

\input{options/algorithmOptions}
\input{options/tikzOptions}

\input{options/customCommands}

\title{A Fixed-Parameter Algorithm for the Max-Cut\\ Problem on Embedded 1-Planar Graphs}

\author{Christine Dahn \and Nils M.~Kriege \and Petra Mutzel}

\institute{Department of Computer Science, TU Dortmund University, Germany\\ 
	\email{\{christine.dahn|nils.kriege|petra.mutzel\}@cs.tu-dortmund.de}}

\begin{document}
\maketitle                

\begin{abstract}
\input{sections/abstract}

\end{abstract}

\section{Introduction}
\input{sections/intro}

\section{Preliminaries}
\label{se:preliminaries}
\input{sections/prelim}

\section{Max-Cut for embedded 1-planar graphs}
\label{se:algkalmost}
\input{sections/mc_1-pl}

\section{Conclusion and open problems}
\label{se:conclusion}
\input{sections/concl}


\bibliographystyle{plain}
\bibliography{literatur}

\end{document}

%% file: options/algorithmOptions.tex
\renewcommand{\algorithmicrequire}{\textbf{Input:}\@\xspace}
\renewcommand{\algorithmicensure}{\textbf{Output:}\@\xspace}


%% file: options/tikzOptions.tex
\tikzset{snake it/.style={decorate, decoration=snake}}
\tikzset{double/.style={blue,densely dotted, very thick}}
\tikzset{inCut/.style={red, decorate, decoration={snake,amplitude=.5mm,segment length=1.5mm,post length=0mm}}}
\tikzset{inCutGray/.style={gray, decorate, decoration={snake,amplitude=.5mm,segment length=1.5mm,post length=0mm}}}
\tikzset{inCutDouble/.style={blue,densely dotted, thick, decorate, decoration={snake,amplitude=.5mm,segment length=1.5mm,post length=0mm}}}



%% file: options/customCommands.tex
\usepackage{xparse}

\newcommand{\sutt}[1]{_{\texttt{#1}}}
\newcommand{\ALG}[2]{\texttt{\textsc{#1}}\sutt{#2}}

\newcommand{\AlgMC}[1]{\ALG{MaxCut}{#1}}

\newcommand{\MC}{\textsc{Max-Cut}}
\newcommand{\MCP}{\textsc{Max-Cut} problem\@\xspace}
\newcommand{\MCA}{\textsc{Max-Cut} algo\-rithm\@\xspace}

\newcommand{\kfplMC}[1]{\AlgMC{#1}}

\newcommand{\contract}{\texttt{\textsc{Update}}}
\newcommand{\Split}{\texttt{\textsc{Split}}}

\newcommand{\argmax}[1]{\underset{#1}{\text{argmax}}\;}

\newcommand{\I}{X}	
\newcommand{\kz}{\chi}	
\newcommand{\wert}{c}	

\newcommand{\nodeCut}{S}	
\newcommand{\nodeKCut}{\overline{\nodeCut}}	
\newcommand{\nodeCutExt}[1]{\nodeCut{#1}}
\newcommand{\nodeoverlineCutExt}[1]{\overline{\nodeCut{#1}}}

\newcommand{\Cut}{\delta(\nodeCut,G)}	
\newcommand{\CutExt}[1]{\delta(\nodeCut{#1})}

\newcommand{\etAl}{et.\nolinebreak[4]\hspace{0.125em}\nolinebreak[4]al.\@\xspace}


%% file: sections/abstract.tex
We propose a fixed-parameter tractable algorithm for the \MCP on embedded 1-planar graphs 
parameterized by the crossing number $k$ of the given embedding.
A graph is called 1-planar if it can be drawn in the plane with at most one crossing per edge.
Our algorithm recursively reduces a 1-planar graph to at most $3^k$ planar graphs, using edge removal and node contraction.
The \MCP is then solved on the planar graphs using established polynomial-time algorithms. 
We show that a maximum cut in the given 1-planar graph can be derived from the solutions for the planar graphs.
Our algorithm computes a maximum cut in an embedded 1-planar graph with $n$ nodes and $k$ edge crossings in time $\mathcal{O}(3^k \cdot n^{3/2} \log n)$.

\keywords{maximum cut, fixed-parameter tractable, 1-planar graphs}

%% file: sections/intro.tex
Partitioning problems on graphs receive increasing attention in the literature. Here the task is to partition the set of nodes of a given (weighted) undirected graph so that the number (or weighted sum) of connections between the parts is minimized.
A special case is the \MCP which asks for a node partition into two sets so that the sum of the edge weights in the cut is maximised.  
The problem is getting increasing attention in the literature, since it is directly related to solving Ising spin glass models (see, e.g., Barahona \cite{Barahona1982}) which are of high interest in physics. Besides its theoretical merits, Ising spin glass models need to be solved in adiabatic quantum computation \cite{McGeoch2014}.
Other applications occur in the layout of electronic circuits \cite{BarahonaGJR88,DDJMRR1995}.

The \MCP has been shown to be NP-hard for general graphs~\cite{Karp72}.
Moreover, Papadimitriou and Yannakakis \cite{PapadimitriouY1991} have shown that the \MCP is APX-hard, i.e., there does not exist a polynomial-time approximation scheme unless P=NP.
Goemans and Williamson suggested a randomized constant factor approximation algorithm  \cite{GoemansW1995} which has been derandomized by Mahajan and Ramesh \cite{MahajanR1999} and has a  performance guarantee of 0.87856.
There are a number of special cases for which the problem can be solved in polynomial time.
The most prominent case arises if the weights of all edges are negative, since then the problem 
can be solved via network flow.
Other special cases are graphs without long odd cycles \cite{GrotschelN84} or weakly bipartite graphs \cite{GroetschelPulleyblank1981}.
Another prominent case arises for planar input graphs. Orlova and Dorfman \cite{OrlovaD72} and Hadlock \cite{Hadlock75} have shown that the \MCP can be solved in polynomial time for unweighted planar graphs.
Their algorithms can be extended to work on weighted planar graphs (e.g., Mutzel \cite{GA1516}).
Currently, the fastest algorithms have been suggested by Shih \etAl\ \cite{ShihWuKuo90} and by Liers and Pardella \cite{LiersP12}.
These results have been extended to the class of graphs not contractible to $K_5$~\cite{Barahona83} and to toroidal graphs \cite{Barahona1983,GalluccioLoebl1998} (i.e., graphs that can be embedded on the torus).
In this paper, we show an extension of the planar special case to that of the class of 1-planar graphs.

A graph is 1-planar if it can be drawn into the plane so that every edge is crossed at most once.
While planarity testing can be done in linear time \cite{HopcroftT73}, the recognition problem for 1-planar graphs is much harder.
Korzhik and Mohar showed that 1-planarity testing is NP-hard \cite{KorzhikMohar13}. 
However, there are fixed-parameter tractable (FPT) algorithms parameterized by the cyclomatic number (the minimum number of edges that must be removed from the graph to create a forest), the tree-depth or the node cover number \cite{BannisterCE13}. For 1-planar graphs these algorithms construct a 1-planar embedding.

\paragraph{Our contribution.}
Given an embedded 1-planar graph with $k$ crossings, we suggest a fixed-parameter tractable algorithm for the \MCP with parameter $k$. The idea of our algorithm is to recursively reduce the input graph into a set of at most $3^k$ planar graphs using a series of edge removals and node contractions. The planar instances can then be solved using the polynomial time algorithms suggested in \cite{ShihWuKuo90,LiersP12} with running time $O(n^{3/2}\log n)$ for a planar graph with $n$ nodes.

The paper is organized as follows. Section \ref{se:preliminaries} contains the basic definitions concerning cuts and 1-planarity.
We also introduce the class of $k$-almost-planar graphs which have 1-planar drawings not exceeding $k$ crossings.
In Section \ref{se:algkalmost} we present our new algorithm for embedded 1-planar graphs and prove its correctness.
Our analysis of its running time shows that it is fixed-parameter tractable with parameter $k$. 
We end with a conclusion and open problems in Section \ref{se:conclusion}.


%% file: sections/prelim.tex
Throughout our paper, we consider undirected weighted graphs $G=(V,E,c)$ with arbitrary edge weights.
A partition of the nodes of $G$ into two sets $S \subseteq V$ and $\overline{S}=V\backslash S$
defines the \emph{cut} $\Cut=\{(uv) \in E \mid (u\in S \hbox{ and }v\in \overline{S}) \hbox{ or } (v\in S \hbox{ and }u\in \overline{S})  \}$.
The \emph{value of a cut} $\Cut$ in the graph $G$ is the sum of weights of all edges in the cut: $\wert(\Cut)= \sum_{e\in \Cut} c_{e}$. 
The \MCP searches for a cut in a given weighted graph with highest value.
For the graph class of planar graphs, the \MCP can be solved in polynomial time.

A graph is \emph{planar} if it admits a \emph{planar drawing}, i.e., a drawing on the plane without any edge crossing.
A drawing admits a \emph{rotation system} which is a clockwise-ordering of the incident edges for every node. In a planar drawing, a rotation system defines the \emph{faces}, i.e., the topologically connected regions of the plane.
One of the faces, the \emph{outer face}, is unbounded. 
A face is uniquely described by its boundary edges. 
Such a description for each face is an equivalent definition of a (planar) embedding. 
A \emph{(planar) embedding} represents the set of all planar drawings with the same faces. It can be represented by the description of the faces or by the rotation system.
It is well known that planarity testing can be solved in linear time \cite{HopcroftT73}. The same is true for computing a planar embedding.
In order to generate crossing free drawings of planar graphs, a number of various algorithms exist, e.g., the straight-line drawing algorithm by de Fraysseix \etAl~\cite{DeFraysseixPP1990}.

Planar graphs are contained in the class of 1-planar graphs. A graph is \emph{1-planar}, if it admits a \emph{1-planar drawing}, i.e., a drawing on the plane with at most one crossing per edge.
Testing 1-planarity is NP-hard \cite{KorzhikMohar13} even in the case of bounded treewidth or bandwith \cite{BannisterCE13}.
A \emph{1-planar embedding} defines the faces of a given 1-planar drawing and can be represented by the set of crossings $X$ and a list of edges and edge segments (half edges) for the crossings for each face.
Note that a 1-planar embedding uniquely defines a rotation system for the nodes. 
However, the opposite is not true. In general, a rotation system does not allow for computing the crossings efficiently or a 1-planar embedding.
Auer et al. \cite{AuerBGR15} have shown that testing 1-planarity of a graph with a fixed rotation system is NP-hard even if the graph is 3-connected.

We call a 1-planar graph \emph{$k$-almost-planar} if it admits a 1-planar drawing with at most $k$ edge crossings.
For edge removal and node contraction we use the following notation: $G - e = (V,E\setminus\{e\})$ denotes the graph obtained from $G=(V,E)$ by deleting the edge $e\in E$. $G/xy$ contracts the two nodes $x$ and $y$ into a new node $v_{xy}\notin V$. In doing so, the edges leading to $x$ or $y$ are replaced by a new edge to $v_{xy}$. Multi-edges to $v_{xy}$ are contracted to one edge and their edge weights are added, self-loops are deleted.
We denote the inverse operation of contraction by $\Split$. 
The contraction and $\Split$ operation can be applied to a subset of nodes $S\subseteq V$:
\begin{align*}
S/xy = & \begin{cases}
S\setminus\{x,y\} \cup \{v_{xy}\} & \text{ if } x,y \in S\\
S &  \text{ otherwise} 
\end{cases}\\
\Split(S,v_{xy}) = & \begin{cases}
S\setminus\{v_{xy}\} \cup \{x,y\} & \text{ if } v_{xy} \in S\\
S & \text{ otherwise}
\end{cases}
\end{align*}

%% file: sections/mc_1-pl.tex
Our main idea for computing the maximum cut in an embedded 1-planar graph $G$ is to eliminate its $k$ crossings and then use a \MCA for planar graphs on the resulting planar graph. In order to remove a crossing, we need to know the two crossing edges of each crossing. 
We use two methods to remove a crossing: Either by deleting one of the crossing edges, or by contracting two nodes that do not belong to the same crossing edge.

\subsection{Removing the crossings}

In this section let $G=(V,E,c)$ be a $k$-almost-planar graph with a 1-planar embedding $(\Pi,\I)$ and a set of crossing edges $\I$ with $\vert \I \vert = k$.
A crossing is defined by a pair of crossing edges, e.g., let $\kz = \{e_{vy}, e_{wz}\} \in \I$ be an arbitrary crossing. 
The following lemma shows that specific node contractions (and edge deletions) remove at least one crossing and do not introduce new crossings.
Figures~\ref{fig:rm_cr_4} and~\ref{fig:rm_cr_3} show examples of node contraction and Figure \ref{fig:rm_cr_2} shows an example of edge deletion.

\begin{lemma}
	\label{lem:k-1-al-pl}
	Let $G$ be a $k$-almost-planar graph with 1-planar embedding $(\Pi,\I)$ and $\kz = \{e_{vy}, e_{wz}\} \in \I$ be an arbitrary crossing.
	The graphs $G/ab$, $G - e_{vy}$ and $G - e_{wz}$ are $(k-1)$-almost-planar for $ab \in \{vw,vz,wy,yz\}$.
	The set of crossings in the resulting 1-planar embedding is a proper subset of $\I$.
\end{lemma}
\begin{proof}
	Since the contracted nodes $a$ and $b$ are each an endpoint to one of the crossing edges, the contracted node is an endpoint to both edges. Since $e_{vy}$ and $e_{wz}$ now have a common endpoint, they can be drawn without a crossing. Therefore the crossing $\kz$ is removed. 
	The contraction does not create new crossings because the two nodes $a$ and $b$ can be moved along their half edges towards the crossing. 
	This is possible because we have a 1-planar embedding which has the property that every crossing is incident to two half edges connecting it with its endpoints. 
	The new node $v_{ab}$ is then placed where the crossing used to be. All other edges can be extended to the new node along the way of the same half edges without creating new crossings. Multi-edges are merged into a single edge and self-loops are deleted.
	In $G - e_{vy}$ and $G - e_{wz}$, the crossing $\kz$ is removed by deleting one of its crossing edges.
	Obviously this does not lead to new crossings.
	So in both cases the number of crossings decreases.
\qed
\end{proof}

\begin{figure}[tb]
	\centering
	\subcaptionbox{$H$
		\label{fig:rm_cr_0}}[.49\linewidth]
	{\resizebox{.44\linewidth}{!}{\input{grafiken/bsp/H}}}
	\subcaptionbox{$H/bc$
		\label{fig:rm_cr_4}}[.49\linewidth]
	{\resizebox{.44\linewidth}{!}{\input{grafiken/bsp/H46}}}
	\subcaptionbox{$H/cd$
		\label{fig:rm_cr_3}}[.49\linewidth]
	{\resizebox{.44\linewidth}{!}{\input{grafiken/bsp/H30}}}
	\subcaptionbox{$H - e_{bd}$
		\label{fig:rm_cr_2}}[.49\linewidth]
	{\resizebox{.44\linewidth}{!}{\input{grafiken/bsp/H20}}}
	\caption{An example how a crossing can be removed. (Blue dashed edges are merged edges from $H$.)}
	\label{fig:rm_cr}
\end{figure}
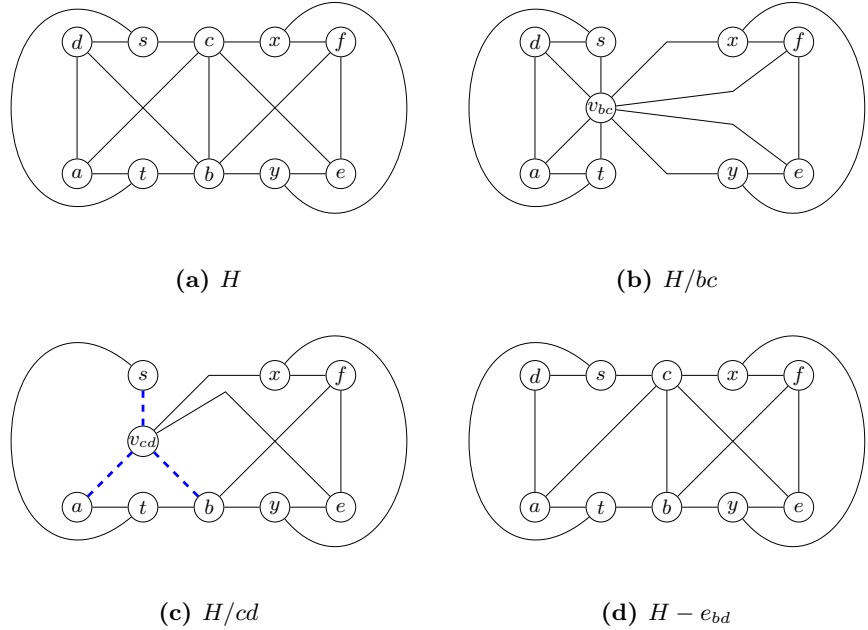

The recursive application of Lemma \ref{lem:k-1-al-pl} shows that all crossings can be removed with these two operations. Thus after $k$ contraction or removal operations, the resulting graph is planar and a planar \MCA can be applied to compute a maximum cut.
The following lemma shows how to project a cut in $G/xy$ or $G - e$ back onto $G$.

\begin{lemma}
\label{lem:cut-transfer-to-G}\ 
Let $G=(V,E,c)$ be a weighted graph, $x\not= y\in V$ and $S\subseteq V$.
\begin{enumerate}[(i)]
	\item Let $\CutExt{,G/xy}$ be a cut in $G/xy$, then the cut $\delta(\Split(\nodeCut,v_{xy}),G)$ in $G$ has the same value. 
	\label{lem:cut-transfer-to-G_contraction}
	\item If $x$ and $y$ are in the same set ($x,y \in \nodeCut$ or $x,y \in \nodeKCut$) then $\Cut = \CutExt{,G - e_{xy}}$ for $e_{xy}\in E$.
	\label{lem:cut-transfer-to-G_removal}
	\end{enumerate}
\end{lemma}
\begin{proof}
	\textit{(i)} Let $\nodeCut$ define a cut in $G/xy$. If the contracted node $v_{xy}$ is split, the cut is projected from $G/xy$ to $G$. The corresponding set of nodes in $G$ is $\nodeCutExt{'}=\Split(\nodeCut,v_{xy})$. It defines a cut in $G$. If $v_{xy} \notin \nodeCut$ then $\nodeCut = \nodeCutExt{'}$.
	The weight of an edge $e \in \CutExt{,G/xy}$ in $G/xy$ is either the same as the weight of the corresponding edge $e' \in \CutExt{',G}$ in $G$ or it is split between two edges $e', e'' \in \CutExt{',G}$ in $G$. 
	The only edge that might exist in $G$ but not in $G/xy$ is $e_{xy}$.
	Since $x$ and $y$ were contracted in $G/xy$, they are either both in $\nodeCutExt{'}$ or both in $\nodeoverlineCutExt{'}$. 
	Therefore the only edge that could be added in $G$ by splitting $v_{xy}$ can not add to the value of $\CutExt{',G}$ in $G$.
	So no weights are lost or added due to the projection and the two cuts have the same value.\\
	\textit{(ii)} This is obvious because $e_{xy}$ is in neither of the two cuts.
\qed 
\end{proof}

\subsection{The Max-Cut~Algorithm}
\label{suse:algo}
We use the three operations introduced above to successively remove all 
crossings of a 1-planar graph. All planar instances obtained in this way are
then solved by a \MC~algorithm for planar graphs. From the solutions of the 
planar graphs, we construct a solution for the original graph.
Note that the algorithm only needs the graph $G$ and the set of edge crossings $\I$ as input.
However, the 1-planar embedding is needed to show the correctness of the algorithm.

Algorithm \ref{algo:MaxCutkfpl} realizes this approach with a recursive function,
which is initially called with the input graph $G$ and the set of crossings $\I$
present in its embedding.
As the algorithm progresses, the graph is successively modified and the set of
crossings is adjusted according to the modifications applied.
If the graph $G$ passed as parameter to the function is planar ($X=\emptyset$), then a planar \MCA is called (line~\ref{algo:MaxCutkfpl:plMC}). 
If there are still crossings remaining, an arbitrary crossing is selected and 
removed in three different ways: 
\begin{algorithm}[tb]
	\input{sections/MCA_1-pl}
	\label{algo:MaxCutkfpl}
\end{algorithm}
Let $y$ be an arbitrary endpoint of one crossing edge and $e_{wz}$, $w \neq y$,
$z \neq y$, the other crossing edge, then (i) the nodes $y$ and $w$ are contracted, 
(ii) the nodes $y$ and $z$ are contracted, and (iii) the edge $e_{wz}$ is deleted.
Each operation removes at least the selected crossing, but in case (i) and (ii) 
also other crossing may be affected. Therefore, the set of crossings $\I$ is 
adjusted by the function $\contract$.
If two nodes $w,y$ are contracted, $\contract(\I,w,y)$ removes every crossing in $\I$ which was dissolved by contracting $w$ and $y$, and replaces every appearance of $w$ or $y$ in $\I$ with the contracted node $v_{wy}$. To check if a crossing was dissolved, $\contract$ checks if $w$ and $y$ are both part of the crossing. Since every crossing needs to be checked once, $\contract$ has a linear running time.
For each case, the recursive function is called with the modified graph and the set of crossings
as a parameter (lines~\ref{algo:MaxCutkfpl:f1}-\ref{algo:MaxCutkfpl:f3}).
Each call returns a node set defining a maximum cut in the modified instance.
The cut with maximal value is then projected back to $G$. If the maximum cut
is obtained in a graph with contracted nodes, i.e., case (i) or (ii), then the
original nodes are restored by the function $\Split(\nodeCut,v_{wy})$, which 
replaces $v_{wy}$ with $w$ and $y$ if $\nodeCut$ contains the contracted node.
This cut-defining set is then returned as the solution to the subproblem.

\begin{figure}[tb]
	\centering
	\subcaptionbox{\MC\ in $H/bc$
		\label{fig:kfplMC_4}}[.49\linewidth]
	{\resizebox{.44\linewidth}{!}{\input{grafiken/bsp/H46_MC}}}
	\subcaptionbox{\MC\ in $H/cd$
		\label{fig:kfplMC_3}}[.49\linewidth]
	{\resizebox{.44\linewidth}{!}{\input{grafiken/bsp/H30_MC}}}
	\subcaptionbox{\MC\ in $H - e_{12}$
		\label{fig:kfplMC_2}}[.49\linewidth]
	{\resizebox{.44\linewidth}{!}{\input{grafiken/bsp/H20_MC}}}
	\subcaptionbox{\MC\ in $H$
		\label{fig:kfplMC_MC}}[.49\linewidth]
	{\resizebox{.44\linewidth}{!}{\input{grafiken/bsp/H_MC}}}
	\caption{An example how the algorithm calculates a \MC~in an embedded 2-almost-planar graph. (Blue dotted edges were merged and have weight 2; all other edges have weight 1; curvy edges belong to the cut; black dashed edges do not belong to the cut.)}
	\label{fig:bsp5}
\end{figure}
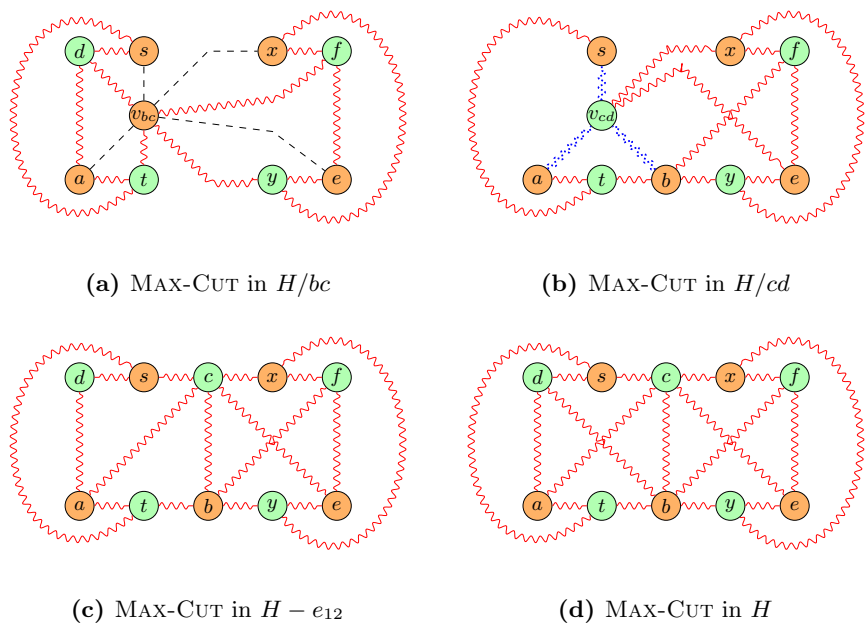
\begin{example}
	\label{bsp:bsp5}
	Given the 2-almost planar graph $H$ in Figure \ref{fig:rm_cr_0} with uniform edge weights, e.g. 1. The algorithm removes the left crossing in the three described ways. The resulting graphs are shown in Figure \ref{fig:rm_cr_4}-\ref{fig:rm_cr_2}. The recursively calculated cuts of these graphs are depicted in Figure \ref{fig:kfplMC_4}-\ref{fig:kfplMC_2} with \ref{fig:kfplMC_3} being the largest cut. This cut is transferred back to $H$ by splitting the contracted node $v_{cd}$. The resulting cut is shown in Figure \ref{fig:kfplMC_MC}. It is a maximum cut in $H$.
\end{example}

\subsection{Correctness}
\label{suse:ana}

The four endpoints of a crossing can be partitioned in eight non-isomorphic ways,
cf.\@ Figure \ref{fig:8sepOptions}: 
(a) all endpoints in one set,
(b)/(c)/(d)/(e) three endpoints in one set without $v$/$w$/$y$/$z$,
(f)/(g) the two endpoints of different crossing edges in the same sets, or
(h) the two endpoints of the same crossing edges in one set each. 
For arbitrary graphs, the induced cut is different because non-crossing edges might be replaced with a path or might not exist at all.

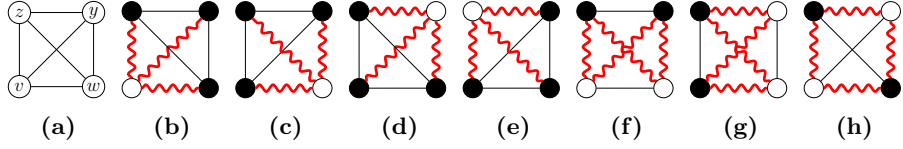
\begin{figure}[tb]
	\centering
	\subcaptionbox{\label{fig:8sepOption0}}[.115\linewidth]
	{\resizebox{.11\linewidth}{!}{\input{grafiken/sepOpt/0}}}
	\subcaptionbox{\label{fig:8sepOption1a}}[.115\linewidth]
	{\resizebox{.11\linewidth}{!}{\input{grafiken/sepOpt/1a}}}
	\subcaptionbox{\label{fig:8sepOption1b}}[.115\linewidth]
	{\resizebox{.11\linewidth}{!}{\input{grafiken/sepOpt/1b}}}
	\subcaptionbox{\label{fig:8sepOption1c}}[.115\linewidth]
	{\resizebox{.11\linewidth}{!}{\input{grafiken/sepOpt/1c}}}
	\subcaptionbox{\label{fig:8sepOption1d}}[.115\linewidth]
	{\resizebox{.11\linewidth}{!}{\input{grafiken/sepOpt/1d}}}
	\subcaptionbox{\label{fig:8sepOption2h}}[.115\linewidth]
	{\resizebox{.11\linewidth}{!}{\input{grafiken/sepOpt/2h}}}
	\subcaptionbox{\label{fig:8sepOption2v}}[.115\linewidth]
	{\resizebox{.11\linewidth}{!}{\input{grafiken/sepOpt/2v}}}
	\subcaptionbox{\label{fig:8sepOption2d}}[.115\linewidth]
	{\resizebox{.11\linewidth}{!}{\input{grafiken/sepOpt/2d}}}
	\caption{
		The 8 non-isomorphic partitions of the four endpoints of a crossing.
		(The red and curvy edges belong to the cut that is defined by the corresponding partition on the $K_4$.)}
	\label{fig:8sepOptions}
\end{figure}

\begin{lemma}
\label{lem:cut-transfer-from-G}	
Let $G=(V,E,c)$ be a 1-planar graph with a 1-planar embedding $(\Pi,\I)$, $S\subseteq V$, and $\kz = \{e_{vy}, e_{wz}\} \in \I$ be an arbitrary crossing.
	\begin{enumerate}[(i)]
		\item If a cut $\Cut$ in $G$ separates the four endpoints of $\kz$ as shown in Fig. \ref{fig:8sepOptions} (a), (b), (c) or (f) then $\nodeCutExt{_2} = \nodeCutExt{ / yz}$ defines a cut in $G/yz$ with the same value.  If $\Cut$ is maximal in $G$ so is $\CutExt{_2, G/yz}$ in $G/yz$.
		\label{lem:cut-transfer-from-G_contraction_yz}
		\item If a cut $\Cut$ in $G$ separates the four endpoints of $\kz$ as shown in Fig. \ref{fig:8sepOptions} (a), (b), (e) or (g) then $\nodeCutExt{_1} = \nodeCutExt{ / wy}$ defines a cut in $G/wy$ with the same value. If $\Cut$ is maximal in $G$ so is $\CutExt{_1, G/wy}$ in $G/wy$.
		\label{lem:cut-transfer-from-G_contraction_wy}
		\item If a cut $\Cut$ in $G$ separates the four endpoints of $\kz$ as shown in Fig. \ref{fig:8sepOptions} (a), (b), (d) or (h) then $\nodeCutExt{_3} = \nodeCut$ defines a cut in $G - e_{wz}$ with the same value. If $\Cut$ is maximal in $G$, so is $\CutExt{_3, G - e_{wz}}$ in $G - e_{wz}$.
		\label{lem:cut-transfer-from-G_removal_wz}
	\end{enumerate}
\end{lemma}
\begin{proof}
	\textit{(i)} Let $\nodeCut$ define a cut in $G$ that separates the endpoints of $\kz$ as shown in Fig. \ref{fig:8sepOptions} (a), (b), (c) or (f).
	By contracting $y$ and $z$, the set of nodes is projected to $G/yz$ and $\CutExt{_2, G/yz}$ is a cut in $G/yz$. 
	The only edge that might have been removed in $G/yz$ does not add to the value of $\Cut$ in $G$ because $y$ and $z$ are not separated by the cut (Fig. \ref{fig:8sepOptions} (a), (b), (c) or (f)). 
	Therefore, the two cuts have the same value in both graphs.
	Let $\nodeCut$ define a maximum cut in $G$ (with the required property). If there was a cut $\CutExt{',G/yz}$ in $G/yz$ larger than $\CutExt{_2, G/yz}$, then $\Split(\nodeCutExt{'},v_{yz})$ would define a cut in $G$ with the same value as $\CutExt{',G}$ (Lemma \ref{lem:cut-transfer-to-G} \ref{lem:cut-transfer-to-G_contraction}), contradicting that $\Cut$ is maximal in $G$.\\
	\textit{(ii)} The proof of the second proposition is analogous to the proof of the first.\\
	\textit{(iii)} Let $\nodeCut$ define a cut in $G$ that separates the endpoints of $\kz$ as shown in Fig. \ref{fig:8sepOptions} (a), (b), (d) or (h).
	Since $G$ and $G - e_{wz}$ have the same set of nodes, $\CutExt{_3,G - e_{wz}}$ is a cut in $G - e_{wz}$ as well.
	We know that $w$ and $z$ are not separated by the cut (Fig. \ref{fig:8sepOptions} (a), (b), (d) or (h)). 
	Therefore the only edge that was removed in $G - e_{wz}$ does not add to the value of the cut in $G$ and the cut has the same value in both graphs.
	Let $\nodeCut$ define a maximum cut in $G$ (with the required property). If there was a cut $\CutExt{',G - e_{wz}}$ in $G - e_{wz}$ larger than $\CutExt{_3,G - e_{wz}}$, then $\CutExt{',G}$ would be a cut in $G$ as well (Lemma \ref{lem:cut-transfer-to-G} \ref{lem:cut-transfer-to-G_removal}), contradicting that $\Cut$ is maximal in $G$.
	\qed
\end{proof}

\begin{theorem}
	\label{theo:opt}
	Algorithm \ref{algo:MaxCutkfpl} computes a maximum cut in a 1-planar graph $G$, given a set of crossing edges $X$ in a 1-planar embedding of $G$.
\end{theorem}
\begin{proof}
	We prove its optimality by induction over $k$. 
	For $k=0$, the given graph is planar. Thus the \MCA for planar graphs calculates a node set defining a maximum cut in $G$.
	Let $\nodeCutExt{^*}$ define a maximum cut in $G$. We show that the cut $\Cut$ defined by the calculated node set $\nodeCut$ is not smaller than $\CutExt{^*,G}$.
	Let $G_1 = G/wy, G_2 = G/yz$ and $G_3 = G - e_{wz}$ be the $(k-1)$-almost-planar graphs (Lemma \ref{lem:k-1-al-pl}) whose cuts $\CutExt{_1, G_1}, \CutExt{_2, G_2}$ and $\CutExt{_3, G_3}$ are calculated recursively by the algorithm. 
	There are 8 possible ways for $\nodeCutExt{^*}$ to separate the four endpoints of $\kz$. These are shown in Fig. \ref{fig:8sepOptions} (a)--(h).
	If the endpoints of $\kz$ are separated as shown in (a), (b), (e) or (g), then $\CutExt{^*, G}$ has the same value as a maximum cut $\CutExt{_1^*, G_1}$ in $G_1$ (Lemma \ref{lem:cut-transfer-from-G} \ref{lem:cut-transfer-from-G_contraction_wy}). Due to the induction hypothesis, $\CutExt{_1, G_1}$ is not smaller than $\CutExt{_1^*, G_1}$. 
	If the endpoints of $\kz$ are separated as shown in (c) or (f), then $\CutExt{^*, G}$ has the same value as a maximum cut $\CutExt{_2^*, G_2}$ in $G_2$ (Lemma \ref{lem:cut-transfer-from-G} \ref{lem:cut-transfer-from-G_contraction_yz}). Due to the induction hypothesis, $\CutExt{_2, G_2}$ is not smaller than $\CutExt{_2^*, G_2}$. 
	If the endpoints of $\kz$ are separated as shown in (d) or (h), then $\CutExt{^*, G}$ has the same value as a maximum cut $\CutExt{_3^*, G_3}$ in $G_3$ (Lemma~\ref{lem:cut-transfer-from-G}~\ref{lem:cut-transfer-from-G_removal_wz}). Due to the induction hypothesis, $\CutExt{_3, G_3}$ is not smaller than $\CutExt{_3^*, G_3}$. 
	The algorithm chooses the node set defining the largest of these three cuts (line \ref{algo:MaxCutkfpl:maxGew}--\ref{algo:MaxCutkfpl:convert_E}) and projects it back to $G$ without changing its value (Lemma \ref{lem:cut-transfer-to-G}). Thus the calculated cut $\Cut$ is not smaller than $\CutExt{_1, G_1}, \CutExt{_2, G_2}$ and $\CutExt{_3, G_3}$.
\qed 
\end{proof}

\subsection{Running time}
\label{suse:time}

Let $n$ be the number of nodes and $m$ be the number of edges of a given graph. 
It is well known that a 1-planar graph has at most $4n-8$ edges \cite{PachT97}. 
For an arbitrary 1-planar drawing, the number of crossings is bounded by $\frac{m}{2}$, since every edge can be crossed at most once and every crossing needs two edges. With the previous observation, we can establish a bound depending on the number of nodes: $k \leq 2n-4$. 

\begin{theorem}
	\label{theo:time_calc_mc_kap}
	Algorithm \ref{algo:MaxCutkfpl} computes a maximum cut in an embedded 1-planar graph with $n$ nodes and $k$ crossings in time $\mathcal{O}(3^k \cdot T_{p}(n))$, where $T_{p}(n)$ is the running time of a planar \MCA. Using the algorithms suggested in \cite{LiersPardella09} or \cite{ShihWuKuo90}, the running time is $\mathcal{O}(3^k \cdot n^{3/2}  \log n)$.
\end{theorem}
\begin{proof}
	Let $T(k,n)$ be the running time of Algorithm \ref{algo:MaxCutkfpl} on an embedded 1-planar graph $G$ with $n$ nodes and $k$ crossings. 
	If $G$ is planar, our algorithm uses a planar \MCA, resulting in $T(0,n) = T_{p}(n)$.
	$\contract$ has a linear running time of $\mathcal{O}(k)$, since every crossing in $\I$ needs to be checked only once. 
	The contractions of $G/wy$ and $G/yz$ take time $\mathcal{O}(n + m)$ and the edge removal $G - e_{wz}$ takes time $\mathcal{O}(m)$. Reversing a contraction on a set of nodes $S_i$ with $\Split$ takes $\vert S_i\vert$ steps, resulting in a running time of $\mathcal{O}(n)$.
	Hence the recursive running time is: 
	$$T(k,n) = 3 \cdot T(k-1,n) + \mathcal{O}(k + n + m)$$
	An induction proof shows that $T(k,n) = 3^k \cdot (T(0,n) + \sum_{i=1}^{k} 3^{-i} \cdot \mathcal{O}(i + n + m))$.
	Since $m$ is bounded by $4n-8$ \cite{PachT97}, $k$ is bounded by $2n-4$ (see above), $i$ is bounded by $k$ and the geometric sum equals a value between 0 and 1, the overall running time is $\mathcal{O}(3^k \cdot (T_{p}(n) + n))$.
	Liers and Pardella \cite{LiersPardella09} or Shih et al.~\cite{ShihWuKuo90} describe a planar \MCA with a running time of $\mathcal{O}(n^{3/2} \cdot \log n)$, resulting in a concrete  running time of $\mathcal{O}(3^k \cdot n^{3/2}  \log n)$ for our algorithm.
\qed
\end{proof}

If the number of crossings $k$ in a 1-planar embedding is fixed, the running time of Algorithm \ref{algo:MaxCutkfpl} is polynomial. However, in an arbitrary 1-planar embedding, $k$ is not fixed and the factor $3^k$ leads to an exponential worst case running time. But we can show that our algorithm is \textit{fixed-parameter tractable} with parameter $k$.
Since its running time can be split into an exponential part, depending only on the parameter $k$, ($3^k$) and a polynomial part in the size of the input graph ($T_{p}(n) + n$), the algorithm is \textit{fixed-parameter tractable} with parameter $k$.
\begin{theorem}
	\label{theo:fpt}
	The \MCP on embedded 1-planar graphs is \textit{fixed-parameter tractable} parameterized by the crossing number $k$ of the given 1-planar embedding.
\end{theorem}

%% file: grafiken/bsp/H.tex
\begin{tikzpicture}
	\node[draw,circle,minimum size=0.45cm, inner sep=0ex] (a) at (0,0) {\small $a$};
	\node[draw,circle,minimum size=0.45cm, inner sep=0ex] (b) at (2,0) {\small $b$};
	\node[draw,circle,minimum size=0.45cm, inner sep=0ex] (c) at (2,2) {\small $c$};
	\node[draw,circle,minimum size=0.45cm, inner sep=0ex] (d) at (0,2) {\small $d$};
	\node[draw,circle,minimum size=0.45cm, inner sep=0ex] (e) at (4,0) {\small $e$};
	\node[draw,circle,minimum size=0.45cm, inner sep=0ex] (f) at (4,2) {\small $f$};

	\node[draw,circle,minimum size=0.45cm, inner sep=0ex] (ab) at (1,0) {\small $t$};
	\node[draw,circle,minimum size=0.45cm, inner sep=0ex] (cd) at (1,2) {\small $s$};
	\node[draw,circle,minimum size=0.45cm, inner sep=0ex] (be) at (3,0) {\small $y$};
	\node[draw,circle,minimum size=0.45cm, inner sep=0ex] (cf) at (3,2) {\small $x$};

	\draw (a) -- (ab);
	\draw (ab) -- (b);
	\draw (a) -- (c);
	\draw (a) -- (d);
	\draw (b) -- (c);
	\draw (b) -- (d);
	\draw (c) -- (cd);
	\draw (cd) -- (d);
	\draw (b) -- (be);
	\draw (be) -- (e);
	\draw (b) -- (f);
	\draw (c) -- (e);
	\draw (c) -- (cf);
	\draw (cf) -- (f);
	\draw (e) -- (f);

	\draw[-] (ab) to [out=-140,in=-90,looseness=1.5] (-1,1);
	\draw[-] (cd) to [out=140,in=90,looseness=1.5] (-1,1);
	\draw[-] (be) to [out=-50,in=-90,looseness=1.5] (5,1);
	\draw[-] (cf) to [out=50,in=90,looseness=1.5] (5,1);
\end{tikzpicture}

%% file: grafiken/bsp/H46.tex
\begin{tikzpicture}
	\node[draw,circle,minimum size=0.45cm, inner sep=0ex] (a) at (0,0) {\small $a$};
	\node[draw,circle,minimum size=0.45cm, inner sep=0ex] (vbc) at (1,1) {\small $v_{bc}$};
	\node[draw,circle,minimum size=0.45cm, inner sep=0ex]  (d) at (0,2) {\small $d$};
	\node[draw,circle,minimum size=0.45cm, inner sep=0ex] (e) at (4,0) {\small $e$};
	\node[draw,circle,minimum size=0.45cm, inner sep=0ex]  (f) at (4,2) {\small $f$};

	\node[draw,circle,minimum size=0.45cm, inner sep=0ex]  (ab) at (1,0) {\small $t$};
	\node[draw,circle,minimum size=0.45cm, inner sep=0ex] (cd) at (1,2) {\small $s$};
	\node[draw,circle,minimum size=0.45cm, inner sep=0ex]  (be) at (3,0) {\small $y$};
	\node[draw,circle,minimum size=0.45cm, inner sep=0ex] (cf) at (3,2) {\small $x$};

	\draw (a) -- (ab);
	\draw (ab) -- (vbc);
	\draw (a) -- (vbc);
	\draw (a) -- (d);
	\draw (vbc) -- (d);
	\draw (vbc) -- (cd);
	\draw (cd) -- (d);
	\draw (vbc) -- (2,0)-- (be);
	\draw (be) -- (e);
	\draw (vbc) -- (3,1.25) -- (f);
	\draw (vbc) -- (3,0.75) -- (e);
	\draw (vbc) -- (2,2)-- (cf);
	\draw (cf) -- (f);
	\draw (e) -- (f);

	\draw[-,] (ab) to [out=-140,in=-90,looseness=1.5] (-1,1) to [out=90,in=140,looseness=1.5] (cd);
	\draw[-,] (be) to [out=-50,in=-90,looseness=1.5] (5,1) to [out=90,in=50,looseness=1.5] (cf);
\end{tikzpicture}
\vspace*{-6pt}

%% file: grafiken/bsp/H30.tex
\begin{tikzpicture}
	\node[draw,circle,minimum size=0.45cm, inner sep=0ex] (a) at (0,0) {\small $a$};
	\node[draw,circle,minimum size=0.45cm, inner sep=0ex] (b) at (2,0) {\small $b$};
	\node[draw,circle,minimum size=0.45cm, inner sep=0ex]  (vcd) at (1,1) {\small $v_{cd}$};
	\node[draw,circle,minimum size=0.45cm, inner sep=0ex] (e) at (4,0) {\small $e$};
	\node[draw,circle,minimum size=0.45cm, inner sep=0ex]  (f) at (4,2) {\small $f$};

	\node[draw,circle,minimum size=0.45cm, inner sep=0ex]  (ab) at (1,0) {\small $t$};
	\node[draw,circle,minimum size=0.45cm, inner sep=0ex] (cd) at (1,2) {\small $s$};
	\node[draw,circle,minimum size=0.45cm, inner sep=0ex]  (be) at (3,0) {\small $y$};
	\node[draw,circle,minimum size=0.45cm, inner sep=0ex] (cf) at (3,2) {\small $x$};

	\draw (a) -- (ab);
	\draw (ab) -- (b);
	\draw[very thick,blue,dashed] (a) -- (vcd);
	\draw[very thick,blue,dashed] (b) -- (vcd);
	\draw[very thick,blue,dashed] (vcd) -- (cd);
	\draw (b) -- (be);
	\draw (be) -- (e);
	\draw (b) -- (f);
	\draw (vcd) -- (2.25,1.75) -- (e);
	\draw (vcd) -- (2,2) -- (cf);
	\draw (cf) -- (f);
	\draw (e) -- (f);

	\draw[-] (ab) to [out=-140,in=-90,looseness=1.5] (-1,1) to [out=90,in=140,looseness=1.5] (cd);
	\draw[-] (be) to [out=-50,in=-90,looseness=1.5] (5,1) to [out=90,in=50,looseness=1.5] (cf);
\end{tikzpicture}
\vspace*{-6pt}

%% file: grafiken/bsp/H20.tex
\begin{tikzpicture}
	\node[draw,circle,minimum size=0.45cm, inner sep=0ex] (a) at (0,0) {\small $a$};
	\node[draw,circle,minimum size=0.45cm, inner sep=0ex] (b) at (2,0) {\small $b$};
	\node[draw,circle,minimum size=0.45cm, inner sep=0ex]  (c) at (2,2) {\small $c$};
	\node[draw,circle,minimum size=0.45cm, inner sep=0ex]  (d) at (0,2) {\small $d$};
	\node[draw,circle,minimum size=0.45cm, inner sep=0ex] (e) at (4,0) {\small $e$};
	\node[draw,circle,minimum size=0.45cm, inner sep=0ex]  (f) at (4,2) {\small $f$};

	\node[draw,circle,minimum size=0.45cm, inner sep=0ex]  (ab) at (1,0) {\small $t$};
	\node[draw,circle,minimum size=0.45cm, inner sep=0ex] (cd) at (1,2) {\small $s$};
	\node[draw,circle,minimum size=0.45cm, inner sep=0ex]  (be) at (3,0) {\small $y$};
	\node[draw,circle,minimum size=0.45cm, inner sep=0ex] (cf) at (3,2) {\small $x$};

	\draw (a) -- (ab);
	\draw (ab) -- (b);
	\draw (a) -- (c);
	\draw (a) -- (d);
	\draw (b) -- (c);
	\draw (c) -- (cd);
	\draw (cd) -- (d);
	\draw (b) -- (be);
	\draw (be) -- (e);
	\draw (b) -- (f);
	\draw (c) -- (e);
	\draw (c) -- (cf);
	\draw (cf) -- (f);
	\draw (e) -- (f);

	\draw[-] (ab) to [out=-140,in=-90,looseness=1.5] (-1,1) to [out=90,in=140,looseness=1.5] (cd);
	\draw[-] (be) to [out=-50,in=-90,looseness=1.5] (5,1) to [out=90,in=50,looseness=1.5] (cf);
\end{tikzpicture}
\vspace*{-6pt}

%% file: sections/MCA_1-pl.tex
\newcommand{\kfplMCb}{\texttt{MaxCut}}
\underline{$\kfplMCb(G,\I)$}\vspace{0.5em}\\
\algorithmicrequire \ An undirected weighted 1-planar graph $G$ and a set of crossing edges $\I$ in a 1-planar embedding of $G$. \\
\algorithmicensure \ A set $\nodeCut \subseteq V_G$ defining a maximum cut $\Cut \subseteq E_G$ in $G$.

\begin{algorithmic}[1]
\IF{$\I = \emptyset$}
\STATE $\nodeCut \gets \kfplMC{planar}(G)$  									\label{algo:MaxCutkfpl:plMC}
\ELSE
\STATE choose an element $\kz\gets\{e_{vy}, e_{wz}\} \in \I$
\STATE $\nodeCutExt{_1} \gets \kfplMCb(G/wy, \contract(\I,w,y))$				\label{algo:MaxCutkfpl:f1}
\STATE $\nodeCutExt{_2} \gets \kfplMCb(G/yz, \contract(\I,y,z))$				\label{algo:MaxCutkfpl:f2}
\STATE $\nodeCutExt{_3} \gets \kfplMCb(G-e_{wz},\I \setminus \{\kz\})$			\label{algo:MaxCutkfpl:f3}
\STATE $G_1 \gets G/wy,$ $G_2 \gets G/yz,$ $G_3 \gets G-e_{wz}$
\STATE $j \gets \argmax{1\leq i\leq 3} \wert(\CutExt{_i,G_i})$					\label{algo:MaxCutkfpl:maxGew}
\IF{$j = 1$}																	\label{algo:MaxCutkfpl:convert_S}
\STATE $\nodeCut \gets \Split(\nodeCutExt{_1},v_{wy})$							\label{algo:MaxCutkfpl:KS1}
\ELSIF{$j = 2$}
\STATE $\nodeCut \gets \Split(\nodeCutExt{_2},v_{yz})$							\label{algo:MaxCutkfpl:KS2}	
\ELSE
\STATE $\nodeCut \gets \nodeCutExt{_3}$
\ENDIF																			\label{algo:MaxCutkfpl:convert_E}
\ENDIF
\STATE \textbf{return} $\nodeCut$												\label{algo:MaxCutkfpl:return}
\end{algorithmic}
\caption{\MCA for embedded 1-planar graphs}

%% file: grafiken/bsp/H46_MC.tex
\begin{tikzpicture}
	\node[draw,circle,minimum size=0.45cm, inner sep=0ex, fill=orange!60] (a) at (0,0) {\small $a$};
	\node[draw,circle,minimum size=0.45cm, inner sep=0ex, fill=orange!60] (vbc) at (1,1) {\small $v_{bc}$};
	\node[draw,circle,minimum size=0.45cm, inner sep=0ex, fill=green!30]  (d) at (0,2) {\small $d$};
	\node[draw,circle,minimum size=0.45cm, inner sep=0ex, fill=orange!60] (e) at (4,0) {\small $e$};
	\node[draw,circle,minimum size=0.45cm, inner sep=0ex, fill=green!30]  (f) at (4,2) {\small $f$};

	\node[draw,circle,minimum size=0.45cm, inner sep=0ex, fill=green!30]  (ab) at (1,0) {\small $t$};
	\node[draw,circle,minimum size=0.45cm, inner sep=0ex, fill=orange!60] (cd) at (1,2) {\small $s$};
	\node[draw,circle,minimum size=0.45cm, inner sep=0ex, fill=green!30]  (be) at (3,0) {\small $y$};
	\node[draw,circle,minimum size=0.45cm, inner sep=0ex, fill=orange!60] (cf) at (3,2) {\small $x$};

	\draw[inCut] (a) -- (ab);
	\draw[inCut] (ab) -- (vbc);
	\draw[dashed] (a) -- (vbc);
	\draw[inCut] (a) -- (d);
	\draw[inCut] (vbc) -- (d);
	\draw[dashed] (vbc) -- (cd);
	\draw[inCut] (cd) -- (d);
	\draw[inCut] (vbc) -- (2,0)-- (be);
	\draw[inCut] (be) -- (e);
	\draw[inCut] (vbc) -- (3,1.25) -- (f);
	\draw[dashed] (vbc) -- (3,0.75) -- (e);
	\draw[dashed] (vbc) -- (2,2)-- (cf);
	\draw[inCut] (cf) -- (f);
	\draw[inCut] (e) -- (f);

	\draw[-,inCut] (ab) to [out=-140,in=-90,looseness=1.5] (-1,1) to [out=90,in=140,looseness=1.5] (cd);
	\draw[-,inCut] (be) to [out=-50,in=-90,looseness=1.5] (5,1) to [out=90,in=50,looseness=1.5] (cf);
\end{tikzpicture}
\vspace*{-6pt}

%% file: grafiken/bsp/H30_MC.tex
\begin{tikzpicture}
	\node[draw,circle,minimum size=0.45cm, inner sep=0ex, fill=orange!60] (a) at (0,0) {\small $a$};
	\node[draw,circle,minimum size=0.45cm, inner sep=0ex, fill=orange!60] (b) at (2,0) {\small $b$};
	\node[draw,circle,minimum size=0.45cm, inner sep=0ex, fill=green!30]  (vcd) at (1,1) {\small $v_{cd}$};
	\node[draw,circle,minimum size=0.45cm, inner sep=0ex, fill=orange!60] (e) at (4,0) {\small $e$};
	\node[draw,circle,minimum size=0.45cm, inner sep=0ex, fill=green!30]  (f) at (4,2) {\small $f$};

	\node[draw,circle,minimum size=0.45cm, inner sep=0ex, fill=green!30]  (ab) at (1,0) {\small $t$};
	\node[draw,circle,minimum size=0.45cm, inner sep=0ex, fill=orange!60] (cd) at (1,2) {\small $s$};
	\node[draw,circle,minimum size=0.45cm, inner sep=0ex, fill=green!30]  (be) at (3,0) {\small $y$};
	\node[draw,circle,minimum size=0.45cm, inner sep=0ex, fill=orange!60] (cf) at (3,2) {\small $x$};

	\draw[inCut] (a) -- (ab);
	\draw[inCut] (ab) -- (b);
	\draw[inCutDouble] (a) -- (vcd);
	\draw[inCutDouble] (b) -- (vcd);
	\draw[inCutDouble] (vcd) -- (cd);
	\draw[inCut] (b) -- (be);
	\draw[inCut] (be) -- (e);
	\draw[inCut] (b) -- (f);
	\draw[inCut] (vcd) -- (2.25,1.75) -- (e);
	\draw[inCut] (vcd) -- (2,2) -- (cf);
	\draw[inCut] (cf) -- (f);
	\draw[inCut] (e) -- (f);

	\draw[-,inCut] (ab) to [out=-140,in=-90,looseness=1.5] (-1,1) to [out=90,in=140,looseness=1.5] (cd);
	\draw[-,inCut] (be) to [out=-50,in=-90,looseness=1.5] (5,1) to [out=90,in=50,looseness=1.5] (cf);
\end{tikzpicture}
\vspace*{-6pt}

%% file: grafiken/bsp/H20_MC.tex
\begin{tikzpicture}
	\node[draw,circle,minimum size=0.45cm, inner sep=0ex, fill=orange!60] (a) at (0,0) {\small $a$};
	\node[draw,circle,minimum size=0.45cm, inner sep=0ex, fill=orange!60] (b) at (2,0) {\small $b$};
	\node[draw,circle,minimum size=0.45cm, inner sep=0ex, fill=green!30]  (c) at (2,2) {\small $c$};
	\node[draw,circle,minimum size=0.45cm, inner sep=0ex, fill=green!30]  (d) at (0,2) {\small $d$};
	\node[draw,circle,minimum size=0.45cm, inner sep=0ex, fill=orange!60] (e) at (4,0) {\small $e$};
	\node[draw,circle,minimum size=0.45cm, inner sep=0ex, fill=green!30]  (f) at (4,2) {\small $f$};

	\node[draw,circle,minimum size=0.45cm, inner sep=0ex, fill=green!30]  (ab) at (1,0) {\small $t$};
	\node[draw,circle,minimum size=0.45cm, inner sep=0ex, fill=orange!60] (cd) at (1,2) {\small $s$};
	\node[draw,circle,minimum size=0.45cm, inner sep=0ex, fill=green!30]  (be) at (3,0) {\small $y$};
	\node[draw,circle,minimum size=0.45cm, inner sep=0ex, fill=orange!60] (cf) at (3,2) {\small $x$};

	\draw[inCut] (a) -- (ab);
	\draw[inCut] (ab) -- (b);
	\draw[inCut] (a) -- (c);
	\draw[inCut] (a) -- (d);
	\draw[inCut] (b) -- (c);
	\draw[inCut] (c) -- (cd);
	\draw[inCut] (cd) -- (d);
	\draw[inCut] (b) -- (be);
	\draw[inCut] (be) -- (e);
	\draw[inCut] (b) -- (f);
	\draw[inCut] (c) -- (e);
	\draw[inCut] (c) -- (cf);
	\draw[inCut] (cf) -- (f);
	\draw[inCut] (e) -- (f);

	\draw[-,inCut] (ab) to [out=-140,in=-90,looseness=1.5] (-1,1) to [out=90,in=140,looseness=1.5] (cd);
	\draw[-,inCut] (be) to [out=-50,in=-90,looseness=1.5] (5,1) to [out=90,in=50,looseness=1.5] (cf);
\end{tikzpicture}
\vspace*{-6pt}

%% file: grafiken/bsp/H_MC.tex
\begin{tikzpicture}
	\node[draw,circle,minimum size=0.45cm, inner sep=0ex, fill=orange!60] (a) at (0,0) {\small $a$};
	\node[draw,circle,minimum size=0.45cm, inner sep=0ex, fill=orange!60] (b) at (2,0) {\small $b$};
	\node[draw,circle,minimum size=0.45cm, inner sep=0ex, fill=green!30]  (c) at (2,2) {\small $c$};
	\node[draw,circle,minimum size=0.45cm, inner sep=0ex, fill=green!30]  (d) at (0,2) {\small $d$};
	\node[draw,circle,minimum size=0.45cm, inner sep=0ex, fill=orange!60] (e) at (4,0) {\small $e$};
	\node[draw,circle,minimum size=0.45cm, inner sep=0ex, fill=green!30]  (f) at (4,2) {\small $f$};

	\node[draw,circle,minimum size=0.45cm, inner sep=0ex, fill=green!30]  (ab) at (1,0) {\small $t$};
	\node[draw,circle,minimum size=0.45cm, inner sep=0ex, fill=orange!60] (cd) at (1,2) {\small $s$};
	\node[draw,circle,minimum size=0.45cm, inner sep=0ex, fill=green!30]  (be) at (3,0) {\small $y$};
	\node[draw,circle,minimum size=0.45cm, inner sep=0ex, fill=orange!60] (cf) at (3,2) {\small $x$};

	\draw[inCut] (a) -- (ab);
	\draw[inCut] (ab) -- (b);
	\draw[inCut] (a) -- (c);
	\draw[inCut] (a) -- (d);
	\draw[inCut] (b) -- (c);
	\draw[inCut] (b) -- (d);
	\draw[inCut] (c) -- (cd);
	\draw[inCut] (cd) -- (d);
	\draw[inCut] (b) -- (be);
	\draw[inCut] (be) -- (e);
	\draw[inCut] (b) -- (f);
	\draw[inCut] (c) -- (e);
	\draw[inCut] (c) -- (cf);
	\draw[inCut] (cf) -- (f);
	\draw[inCut] (e) -- (f);

	\draw[-,inCut] (ab) to [out=-140,in=-90,looseness=1.5] (-1,1) to [out=90,in=140,looseness=1.5] (cd);
	\draw[-,inCut] (be) to [out=-50,in=-90,looseness=1.5] (5,1) to [out=90,in=50,looseness=1.5] (cf);
\end{tikzpicture}
\vspace*{-6pt}

%% file: grafiken/sepOpt/0.tex
\begin{tikzpicture}
	\node[draw,circle,minimum size=0.6cm, inner sep=0ex] (a) at (0,0) {\Large $v$};
	\node[draw,circle,minimum size=0.6cm, inner sep=0ex] (b) at (2,0) {\Large $w$};
	\node[draw,circle,minimum size=0.6cm, inner sep=0ex] (c) at (2,2) {\Large $y$};
	\node[draw,circle,minimum size=0.6cm, inner sep=0ex] (d) at (0,2) {\Large $z$};
	
	\draw (a) -- (b);
	\draw (a) -- (c);
	\draw (a) -- (d);
	\draw (b) -- (c);
	\draw (b) -- (d);
	\draw (c) -- (d);
\end{tikzpicture}

%% file: grafiken/sepOpt/1a.tex
\begin{tikzpicture}
	\node[draw,circle,minimum size=0.5cm, inner sep=0ex] (a) at (0,0) {};
	\node[draw,fill,circle,minimum size=0.5cm, inner sep=0ex] (b) at (2,0) {};
	\node[draw,fill,circle,minimum size=0.5cm, inner sep=0ex] (c) at (2,2) {};
	\node[draw,fill,circle,minimum size=0.5cm, inner sep=0ex] (d) at (0,2) {};
	
	\draw[red,line width=2pt,snake it] (a) -- (b);
	\draw[red,line width=2pt,snake it] (a) -- (c);
	\draw[red,line width=2pt,snake it] (a) -- (d);
	\draw (b) -- (c);
	\draw (b) -- (d);
	\draw (c) -- (d);
\end{tikzpicture}

%% file: grafiken/sepOpt/1b.tex
\begin{tikzpicture}
	\node[draw,fill,circle,minimum size=0.5cm, inner sep=0ex] (a) at (0,0) {};
	\node[draw,circle,minimum size=0.5cm, inner sep=0ex] (b) at (2,0) {};
	\node[draw,fill,circle,minimum size=0.5cm, inner sep=0ex] (c) at (2,2) {};
	\node[draw,fill,circle,minimum size=0.5cm, inner sep=0ex] (d) at (0,2) {};
	
	\draw[red,line width=2pt,snake it] (a) -- (b);
	\draw (a) -- (c);
	\draw (a) -- (d);
	\draw[red,line width=2pt,snake it] (b) -- (c);
	\draw[red,line width=2pt,snake it] (b) -- (d);
	\draw (c) -- (d);
\end{tikzpicture}

%% file: grafiken/sepOpt/1c.tex
\begin{tikzpicture}
	\node[draw,fill,circle,minimum size=0.5cm, inner sep=0ex] (a) at (0,0) {};
	\node[draw,fill,circle,minimum size=0.5cm, inner sep=0ex] (b) at (2,0) {};
	\node[draw,circle,minimum size=0.5cm, inner sep=0ex] (c) at (2,2) {};
	\node[draw,fill,circle,minimum size=0.5cm, inner sep=0ex] (d) at (0,2) {};

	\draw (a) -- (b);
	\draw[red,line width=2pt,snake it] (a) -- (c);
	\draw (a) -- (d);
	\draw[red,line width=2pt,snake it] (b) -- (c);
	\draw (b) -- (d);
	\draw[red,line width=2pt,snake it] (c) -- (d);
\end{tikzpicture}

%% file: grafiken/sepOpt/1d.tex
\begin{tikzpicture}
	\node[draw,fill,circle,minimum size=0.5cm, inner sep=0ex] (a) at (0,0) {};
	\node[draw,fill,circle,minimum size=0.5cm, inner sep=0ex] (b) at (2,0) {};
	\node[draw,fill,circle,minimum size=0.5cm, inner sep=0ex] (c) at (2,2) {};
	\node[draw,circle,minimum size=0.5cm, inner sep=0ex] (d) at (0,2) {};

	\draw (a) -- (b);
	\draw (a) -- (c);
	\draw[red,line width=2pt,snake it] (a) -- (d);
	\draw (b) -- (c);
	\draw[red,line width=2pt,snake it] (b) -- (d);
	\draw[red,line width=2pt,snake it] (c) -- (d);
\end{tikzpicture}

%% file: grafiken/sepOpt/2h.tex
\begin{tikzpicture}
	\node[draw,circle,minimum size=0.5cm, inner sep=0ex] (a) at (0,0) {};
	\node[draw,circle,minimum size=0.5cm, inner sep=0ex] (b) at (2,0) {};
	\node[draw,fill,circle,minimum size=0.5cm, inner sep=0ex] (c) at (2,2) {};
	\node[draw,fill,circle,minimum size=0.5cm, inner sep=0ex] (d) at (0,2) {};

	\draw (a) -- (b);
	\draw[red,line width=2pt,snake it] (a) -- (c);
	\draw[red,line width=2pt,snake it] (a) -- (d);
	\draw[red,line width=2pt,snake it] (b) -- (c);
	\draw[red,line width=2pt,snake it] (b) -- (d);
	\draw (c) -- (d);
\end{tikzpicture}

%% file: grafiken/sepOpt/2v.tex
\begin{tikzpicture}
	\node[draw,fill,circle,minimum size=0.5cm, inner sep=0ex] (a) at (0,0) {};
	\node[draw,circle,minimum size=0.5cm, inner sep=0ex] (b) at (2,0) {};
	\node[draw,circle,minimum size=0.5cm, inner sep=0ex] (c) at (2,2) {};
	\node[draw,fill,circle,minimum size=0.5cm, inner sep=0ex] (d) at (0,2) {};

	\draw[red,line width=2pt,snake it] (a) -- (b);
	\draw[red,line width=2pt,snake it] (a) -- (c);
	\draw (a) -- (d);
	\draw (b) -- (c);
	\draw[red,line width=2pt,snake it] (b) -- (d);
	\draw[red,line width=2pt,snake it] (c) -- (d);
\end{tikzpicture}

%% file: grafiken/sepOpt/2d.tex
\begin{tikzpicture}
	\node[draw,circle,minimum size=0.5cm, inner sep=0ex] (a) at (0,0) {};
	\node[draw,fill,circle,minimum size=0.5cm, inner sep=0ex] (b) at (2,0) {};
	\node[draw,circle,minimum size=0.5cm, inner sep=0ex] (c) at (2,2) {};
	\node[draw,fill,circle,minimum size=0.5cm, inner sep=0ex] (d) at (0,2) {};

	\draw[red,line width=2pt,snake it] (a) -- (b);
	\draw (a) -- (c);
	\draw[red,line width=2pt,snake it] (a) -- (d);
	\draw[red,line width=2pt,snake it] (b) -- (c);
	\draw (b) -- (d);
	\draw[red,line width=2pt,snake it] (c) -- (d);
\end{tikzpicture}

%% file: sections/concl.tex
We have presented a polynomial time algorithm for computing a \MC~in a 1-planar graph provided with a 1-planar embedding with a constant number of crossings. This shows that the \MCP on embedded 1-planar graphs is in the class FPT. 

The question arises if our approach can be extended to general graphs with up to $k$ crossings per edge, so called \emph{$k$-planar graphs}.
Our approach is based on the fact that node contractions and edge deletions decrease the number of crossings (see Lemma \ref{lem:k-1-al-pl}). 
Figure \ref{fig:k33} shows that this is no longer true if an edge is crossed more than once.
In this case, there are crossings that do not have direct half edges connecting it to its endpoints like, e.g.,  the crossing $(ad,cf)$ in Figure \ref{fig:k33}. If we contract $d$ and $f$, we get plenty of new crossings in the new graph $G/df$.
We are currently working to generalize our approach to embedded $k$-planar graphs.

\begin{figure}[h]
	\centering
	\subcaptionbox{$G$
		\label{fig:4-pl}}[.49\linewidth]
	{\resizebox{.4\linewidth}{!}{\input{grafiken/k_3,3.tex}}}
	\subcaptionbox{$G/df$	\label{fig:4-pl_contr}}[.49\linewidth]
	{\resizebox{.4\linewidth}{!}{\input{grafiken/2-planar_contracted}}}
	\caption{A 4-planar graph where the contraction of the nodes $d$ and $f$ leads to $\mathcal{O}(l)$ new crossings. The two edges that generate the new crossings are drawn in red.
		Between $a$ and $f$ (resp. $c$ and $d$) in $G$ are $l$ independent paths. Beneath $b$ there are $l$ paths between $a$ and $c$ that are pairwise connected and therefore have a specific order. The highest path contains a node connected to $b$ and the lowest path contains a node connected to $e$. No matter where $e$ is drawn in $G/df$, one of the two red edges crosses at least $l-1$ other edges.}
	\label{fig:k33}
\end{figure}
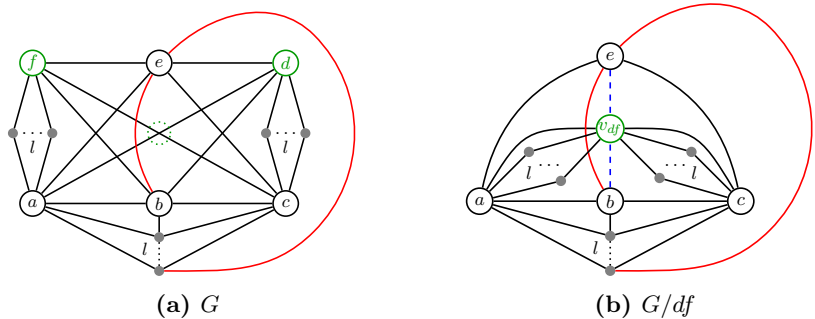

Another interesting question would be to drop the assumption that we are given a 1-planar embedding. 
Note that our algorithm does not need such an embedding as input, it only needs to get a list of edge crossings that must correspond to a 1-planar embedding. 
However, for our correctness analysis it is important to have a 1-planar embedding of the graph.

%% file: grafiken/k_3,3.tex
\begin{tikzpicture}
\begin{scope} [thick] 
		\node[draw,circle,minimum size=0.45cm, inner sep=0ex] (a) at (0   ,0  ) {\small $a$};
		\node[draw,circle,minimum size=0.15cm,fill,gray, inner sep=0ex] (af1) at (-0.35 ,1.25) {};
		\node (afk) at (0   ,1.25) {\small $\ldots$};
		\node (afk) at (0 ,1) {\small $l$};
		\node[draw,circle,minimum size=0.15cm,fill,gray, inner sep=0ex] (afn) at (0.35 ,1.25) {};
		\node[draw,circle,minimum size=0.45cm, inner sep=0ex,green!60!black] (f) at (0   ,2.5) {\small $f$};
		
		\node[draw,circle,minimum size=0.45cm, inner sep=0ex] (b) at (2.25,0  ) {\small $b$};
		\node[draw,circle,minimum size=0.45cm, inner sep=0ex] (e) at (2.25,2.5) {\small $e$};
		
		\node[draw,circle,minimum size=0.45cm, inner sep=0ex] (c) at (4.5 ,0  ) {\small $c$};		
		\node[draw,circle,minimum size=0.15cm,fill,gray, inner sep=0ex] (cd1) at (4.15 ,1.25) {};
		\node (cdk) at (4.5 ,1.25) {\small $\ldots$};
		\node (cdk) at (4.5 ,1) {\small $l$};
		\node[draw,circle,minimum size=0.15cm,fill,gray, inner sep=0ex] (cdn) at (4.85 ,1.25) {};
		\node[draw,circle,minimum size=0.45cm, inner sep=0ex,green!60!black] (d) at (4.5 ,2.5) {\small $d$};

		
		\node[draw,circle,minimum size=0.15cm,fill,gray, inner sep=0ex] (ac1) at (2.25 ,-0.6) {};
		\node[draw,circle,minimum size=0.15cm,fill,gray, inner sep=0ex] (acn) at (2.25 ,-1.2) {};
		\node (ack) at (2,-0.8) {\small $l$};

	\draw (a) -- (d);
	\draw (a) -- (e);
	\draw (b) -- (d);
	\draw (b) -- (f);
	\draw (c) -- (e);
	\draw (c) -- (f);
	
	\draw[-,red] (b) to[bend left] (e);
	 	
	\draw (e) -- (f);
	\draw (e) -- (d);
	\draw (b) -- (a);
	\draw (b) -- (c);
	
	
	\draw (a) -- (af1) -- (f);
	\draw (a) -- (afn) -- (f);
	
	\draw (c) -- (cd1) -- (d);
	\draw (c) -- (cdn) -- (d);
	
	\draw (a) -- (ac1) -- (c);
	\draw (a) -- (acn) -- (c);
	
	\draw[-] (b) to (ac1);
	\draw[-,dotted] (ac1) to (acn);
	
	\draw[-,red] (e) to[out=45,in=135] (5,3) to[out=-45,in=45] (5,-0.5) to[out=-135,in = 0] (acn);
	
	\node[draw,circle,minimum size=0.4cm, inner sep=0ex,green!60!black, dotted] (f) at (2.25,1.25) {};
\end{scope}
\end{tikzpicture}

%% file: grafiken/2-planar_contracted.tex
\begin{tikzpicture}
\begin{scope} [thick] 
		\node[draw,circle,minimum size=0.45cm, inner sep=0ex] (a) at (0   ,0  ) {\small $a$};
		\node[draw,circle,minimum size=0.15cm,fill,gray, inner sep=0ex] (af1) at (0.85 ,0.85) {};
		\node (afk) at (1.2 ,0.625) {\small $\ldots$};
		\node (afk) at (0.85 ,0.5) {\small $l$};
		\node[draw,circle,minimum size=0.15cm,fill,gray, inner sep=0ex] (afn) at (1.4 ,0.35) {};
		\node[draw,circle,minimum size=0.45cm, inner sep=0ex,green!60!black] (vdf) at (2.25,1.25) {\small $v_{df}$};
		
		\node[draw,circle,minimum size=0.45cm, inner sep=0ex] (b) at (2.25,0  ) {\small $b$};
		\node[draw,circle,minimum size=0.45cm, inner sep=0ex] (e) at (2.25,2.5) {\small $e$};
		
		\node[draw,circle,minimum size=0.45cm, inner sep=0ex] (c) at (4.5 ,0  ) {\small $c$};		
		\node[draw,circle,minimum size=0.15cm,fill,gray, inner sep=0ex] (cd1) at (3.65 ,0.85) {};
		\node (cdk) at (3.375 ,0.625) {\small $\ldots$};
		\node (afk) at (3.7 ,0.5) {\small $l$};
		\node[draw,circle,minimum size=0.15cm,fill,gray, inner sep=0ex] (cdn) at (3.1 ,0.4) {};

		
		\node[draw,circle,minimum size=0.15cm,fill,gray, inner sep=0ex] (ac1) at (2.25 ,-0.6) {};
		\node (ack) at (2,-0.8) {\small $l$};
		\node[draw,circle,minimum size=0.15cm,fill,gray, inner sep=0ex] (acn) at (2.25 ,-1.2) {};

	\draw[-] (a) to[bend left, looseness=1.4] (vdf);
	\draw[-] (a) to[bend left] (e);
	\draw[-,dashed, blue, thick] (b) to (vdf);
	\draw[-] (c) to[bend right] (e);
	\draw[-] (c)  to[bend right, looseness=1.4] (vdf);
	
	\draw[-,red] (b) to[bend left] (e);
	 	
	\draw[-,dashed, blue, thick] (e) to (vdf);
	\draw (b) -- (a);
	\draw (b) -- (c);
	
	
	\draw (a) -- (af1) -- (vdf);
	\draw (a) -- (afn) -- (vdf);
	
	\draw[-] (c) to (cd1) to (vdf);
	\draw[-] (c) to (cdn) to (vdf);
	
	\draw (a) -- (ac1) -- (c);
	\draw (a) -- (acn) -- (c);
	
	\draw[-] (b) to (ac1);
	\draw[-,dotted] (ac1) to (acn);
	
	\draw[-,red] (e) to[out=45,in=135] (5,3) to[out=-45,in=45] (5,-0.5) to[out=-135,in = 0] (acn);
	
	
\end{scope}
\end{tikzpicture}